\documentclass[runningheads]{llncs}
\usepackage[T1]{fontenc}
\usepackage{graphicx}

\usepackage[american]{babel}

\usepackage{amsmath}
\usepackage[hidelinks]{hyperref}
\usepackage{cleveref}

\usepackage{enumitem}

\usepackage{afterpage}
\usepackage{lcg}

\usepackage{tikz}
\usepackage{tikzscale}

\usepackage{algorithm, listings}
\usepackage{fca}
\newcommand{\K}{\mathbb{K}}
\newcommand{\R}{\mathbb{R}}
\newcommand{\E}{\mathbb{E}}
\newcommand{\I}{\mathcal{I}}
\newcommand{\J}{\mathcal{J}}

\newcommand{\B}{\mathfrak{B}}

\let\subset\subseteq
\let\supset\supseteq
\let\cref\Cref

\usepackage{color}

\begin{document}
\title{Realizability of Rectangular Euler Diagrams}
\author{Dominik Dürrschnabel\inst{1,2}\orcidID{0000-0002-0855-4185} \and
    Uta Priss\inst{3}\orcidID{0000-0003-0375-429X}}
\authorrunning{Dürrschnabel and Priss}
\institute{Knowledge \& Data Engineering Group,
    University of Kassel, Kassel,  Germany\\
    \and
    Interdisciplinary Research Center for Information System Design, Kassel, Germany\\
    \email{duerrschnabel@cs.uni-kassel.de}
    \and
    Faculty of Computer Science, Ostfalia University of Applied Sciences, Wolfenbüttel, Germany\\
    \email{u.priss@ostfalia.de}}
\maketitle              %
\begin{abstract}
    Euler diagrams are a tool for the graphical representation of set relations.
    Due to their simple way of visualizing elements in the sets by geometric containment, they are easily readable by an inexperienced reader.
    Euler diagrams where the sets are visualized as aligned rectangles are of special interest.
    In this work, we link the existence of such rectangular Euler diagrams to the order dimension of an associated order relation.
    For this, we consider Euler diagrams in one and two dimensions.
    In the one-dimensional case, this correspondence provides us with a polynomial-time  algorithm to compute the Euler diagrams, while the two-dimensional case is linked to an NP-complete problem which we approach with an exponential-time algorithm.
\end{abstract}

\section{Introduction}

Set diagrams are commonly used for the visualization of sets  in set theory.
Thereby, two kinds of set diagrams have proven themselves to be useful.
Venn diagrams, named after John Venn (1834-1923), can be used to visualize elementary set theory by illustrating the sets as geometric objects.
Hereby, frequently, the shapes of circles or ellipsis are used.
They show all possible logical relations as intersections of their regions, i.e., in the case of $n$ sets there are $2^n$ different intersecting zones.
Because of this exponential nature, they reach the limit of what is perceived as readable for values of $n$ as small as 4 or 5.
Euler diagrams however are not subject to  the same restriction.
Only regions which contain a common object have to intersect.
This additional degree of freedom makes them much more suitable for the representation of more complex set relations.
Diagrams of this kind were first introduced by Leonhard Euler (1707-1783).
As a second interpretation, they permit the visualization of binary datasets, such as formal contexts from the research realm of formal concept analysis.
Each set corresponds to an attribute and each element in the sets to an object.
The incidence between objects and attributes is depicted by containment.

\afterpage{
    \setcounter{footnote}{-1}
    \begin{figure}[t]
        \centering
        \includegraphics[width=.7\textwidth]{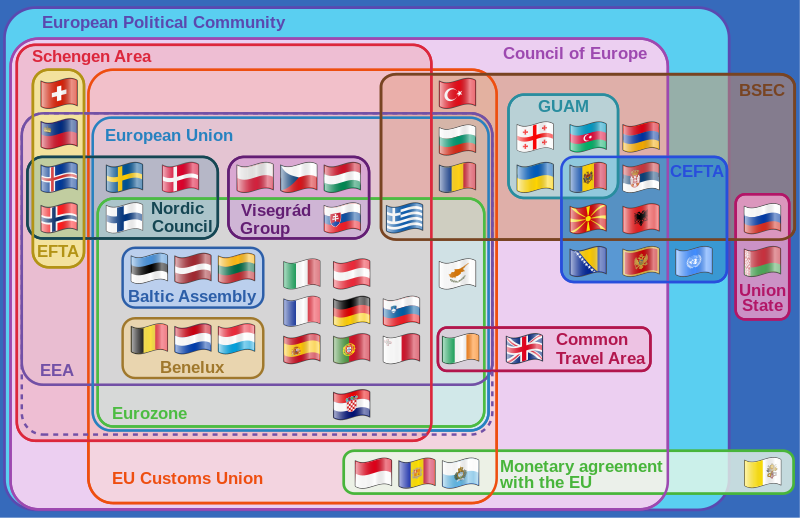}
        \caption{Supranational European bodies represented by a two-dimensional Euler diagram.\protect\footnotemark}
        \label{fig:europa}
    \end{figure}
    \footnotetext{\url{https://commons.wikimedia.org/wiki/File:Supranational_European_Bodies.svg}.\\Accessed January 19, 2024.}
}
One especially readable type of Euler diagrams are the ones, where the sets are represented by aligned rectangles.
Because of their readability, those diagrams are often used to depict set relations without further explanation on how to read them.
One common example is the membership of countries in supranational bodies, as one can see for example in \cref{fig:europa}.
Here, European countries are related to their supranational bodies.
To see that this diagram is in fact an Euler and not a Venn diagram, note that, for example, Eurozone and EFTA do not intersect.

Even though these diagrams are commonly used in applications, their automatic generation is surprisingly not very well investigated.
Until today, most rectangular Euler diagrams are still hand-crafted.
To our knowledge, the first full  research paper dealing with the automatic generation of rectangular Euler diagrams is due to Paetzold et al.\ from 2023~\cite{Paetzold.2023}.
There, a mixed-integer optimization scheme is used to compute the diagrams.
As it is not always possible to represent each family of sets using a rectangular Euler diagram, they iteratively split the diagram into multiple subdiagrams if they cannot find a suitable solution.
As they do not have a reliable way to check for the existence of diagrams for a given dataset, they do a split if their algorithm does not find a solution within 50 seconds.

In this work, we give a criterion for the existence of rectangular Euler diagrams, using the mathematical theory of order dimension.
In doing so, we consider the problem in two different settings.
First, we consider the problem in one dimension.
We show that an associated order relation, which we introduce in this work, has order dimension two if and only if the dataset has a one-dimensional Euler diagram.
This characterization is given in a constructive way, which allows it to be directly translated into a polynomial time algorithm.
Then, we consider the two-dimensional setting, which is the same setting investigated by Paetzold et al.~\cite{Paetzold.2023} and depicted in \cref{fig:europa}.
Similarly to the one-dimensional case, we construct an associated partial order and use it to give a criterion for the exis\-tence of the Euler diagrams based on its order dimension.
In comparison to the one-dimensional case however, this is based on the order dimension being  four, which is NP-complete to decide.
Thus, the algorithm that follows from this condition does not run in polynomial time.

The rest of this work is structured as follows.
In \cref{sec:related_work} we embed our work in previous research.
Then, we introduce the necessary mathematical notations in \cref{sec:math}.
In \cref{sec:1d,sec:2d} we investigate the one- and two-dimensional Euler diagrams respectively, before we discuss the time complexity of the resulting algorithms in \cref{sec:alg}.
Finally, in \cref{sec:concl} we conclude this work.

\section{Related Work}
\label{sec:related_work}

A (partially) ordered set $(X,\leq)$ is called a geometric containment order if there is a mapping from its elements into shapes of the same type in a Euclidean space of finite dimension.
For each pair of elements it holds that $x < y$ if and only if the shape corresponding to $x$ is properly included in the shape of $y$.
Geometric containment orders are closely related to this work.
In our work,  attributes form an interval or rectangle type of geometric containment order.
Furthermore, the objects are represented as points in the same space and the incidence between objects and attributes is represented by geometric containment.
The first result relevant to our work in that regard can be found in the very first paper introducing order dimension~\cite{Dushnik.1941}.
In this work, Dushnik and Miller note that interval containment orders are exactly the orders of dimension of at most two.
A gen\-er\-al\-i\-za\-tion of this result is given in the work of Golumbic and Scheinerman~\cite{Golumbic.1989}, where they show that orders of dimension of at most $2k$ can be characterized as the geometric containment orders with aligned boxes in an at most $k$-dimensional Euclidean space.
Following this, we relate the one-dimensional Euler diagrams to orders of dimension two, and the two-dimensional Euler diagrams to orders of dimension four.
More results on the order dimension of geometric containment orders of various types can be found in a survey by Fishburn and Trotter~\cite{Fishburn.1998}.

For a general overview on Euler diagrams, we refer the reader to a survey of Rodgers~\cite{Rodgers.2014}.
There, Euler diagrams are introduced using the very common definition of being ``finite sets of labeled and closed curves''.
Thus, the Euclidean space is divided into regions or zones by the curves.
This is in contrast to the definition used in this work where we refer to Euler diagrams as  intervals and points on the real line or as rectangles and points in the two-dimensional Euclidean space.
Our reason to deviate from the common notion is that it is easier to work with our definition  in the context of geometric containment.
The only setting which is similar to ours and was previously investigated is  the work of Paetzold et al.~\cite{Paetzold.2023}.
Other approaches that incorporate rectangles in the generated Euler diagrams are given by Riche and Dwyer~\cite{Riche.2010} and Yoghourdjian~\cite{Yoghourdjian.2016}.
Furthermore, there are adjacent approaches for the generation of Euler diagrams if we forgo the condition of the sets being rectangular.
Thereby, an often-applied strategy~\cite{Rodgers.2008,Kehlbeck.2022} is to use the dual graph of the diagram.
This graph contains a vertex for each intersection of regions and two vertices are connected by an edge if their corresponding regions are neighboring.

The canonical way to represent a formal context in formal concept analysis is the concept lattice, for which drawing approaches based on its order dimension exist~\cite{Durrschnabel.2023}.
An alternative approach was developed under the name ordinal factor analysis~\cite{Ganter.2012}.
For this approach several visualization methods based on the order dimension have been proposed recently~\cite{Durrschnabel.2023.2,Durrschnabel2023.3}.
In his famous theorem~\cite{schnyder1989planar}, Schnyder gives a condition on the existence of a planar graph based on the order dimension of its incidence poset.
From this, also an algorithm to compute a plane embedding follows, thus our work is conceptually similar to his.
One especially noteworthy result related to our work is due to Petersen~\cite{Petersen.2004,Petersen.2008} who investigates $S$-alphabets which coincide with our one-dimensional Euler diagrams.
Her work relates the realizability as an $S$-alphabet to the planarity of an associated lattice.
As for lattices planarity and two-dimensionality are equivalent, our work can be seen as an extension of hers.
In fact, the lattice proposed in her work is the Dedekind-MacNeille completion of the order relation discussed in \cref{sec:1d} on one-dimensional Euler diagrams of this work.
In \cite{Priss24}, the authors of this paper describe connections between order theory and Euler diagrams and give some negative results on the relation between interval containment orders and one-dimensional Euler diagrams.
Finally, in the realm of formal concept analysis, there is recent interest in using formal concepts for the generation of Euler diagrams~\cite{Priss.2023}.

\section{Mathematical Foundations}
\label{sec:math}

For binary relations $R$ the notions $(x,y)\in R$ and $xRy$ can be used interchangeably.
A homogeneous relation $\leq$ on a set $X$ that is reflexive ($\forall x \in X : x \leq x$), antisymmetric ($\forall x,y \in X:x \leq y \wedge x \leq y \Rightarrow x = y$), and transitive ($\forall x, y,z \in X : x \leq y \wedge y \leq z \Rightarrow x \leq z$) is called an \emph{order relation} on $X$. We call the tuple $(X, \leq)$ an \emph{ordered set}.
If $x \leq y$ and $x \neq y$ we write $x < y$.
Furthermore, we sometimes use the notation $y \geq x$ instead of $x \leq y$ and $y > x$ instead of $x < y$.
Two elements $x,y \in X$ are called \emph{comparable} if $x \leq y$ or $y \leq x$, otherwise they are called \emph{incomparable}.
A \emph{linear order} on $X$ is an order relation, where all pairs of elements in $X$ are comparable.
For an ordered set $(X, \leq)$ a linear order $\leq_l$ on $X$ is called a \emph{linear extension} of $\leq$, if $\leq_l \supset \leq$.
A family $\mathcal{R}=({\leq_1},\ldots,{\leq_t})$ of linear extensions of $\leq$ is called a \emph{realizer} of $\leq$ if $\leq = \bigcap_{i=1}^t{\leq_i}$.
The \emph{order dimension} of $\leq$ is defined as the least cardinality of a realizer and denoted by $\dim(X,{\leq})$.

A formal context $\K$ is a triple $\K=(G,M,I)$ where $G$ is a set called \emph{objects}, $M$ a set called \emph{attributes}, and $I\subset G \times M$ an \emph{incidence relation}.
For a set $A\subset G$, its \emph{derivation} is given by $A'=\{(m \in M \mid \forall g \in A : (g,m) \in I)\}$. Dually, for $B \subset M$, let $B'=\{g \in G \mid \forall m \in B : (g,m) \in I\}$.
A \emph{formal concept} is a tuple $(A,B)$ with $A \subset G$ and $B \subset M$ such that $A' = B$ and $B' = A$. We call $A$ the \emph{extent} and $B$ the \emph{intent} of the concept, and denote the set of all concepts by $\B(K)$.
The ordered set $(\B,\leq)$ with $(A,B) \leq (C,D)$ iff $A \subset B$ is called the \emph{concept lattice} of $\K$.
In this work, we assume all contexts to be \emph{clarified}, i.e., if multiple attributes share the same derivation, they are grouped to a single attribute, and, if multiple objects share the same derivation, they are grouped to a single object.
Such an assumption can be made, as going back from a clarified to the original context is only a matter of duplication.
We denote by $\I$ the set of all closed intervals in the real numbers $\R$.
A \emph{rectangle} is a tuple $(x,y)$ where $x,y \subset \R$ are closed intervals.
We say that a point $(a,b) \in R \times R$ is \emph{contained} in a rectangle $(x,y)$ if $a \in x$ and $b \in y$.

\section{One-Dimensional Euler Diagrams}
\label{sec:1d}

The first kind of Euler diagram that we are investigating in this work are \emph{one-dimensional}, which means that their sets can be aligned along the real line.
These diagrams are the essential building blocks for the more complex two-dimensional variant that we will discuss in the next section.
In these diagrams, attributes are represented by intervals and objects by points.
Formally, they are defined as follows.

\begin{definition}
    A \emph{one-dimensional Euler diagram} $\E=(\J, P)$ is a set of closed intervals  $\J \subset \I$ and a set of points $P \subset \R$.
    Let $\K=(G,M,I)$ be a formal context.
    We say that $\mathbb{E}$ \emph{corresponds} to $\K$ with bijective maps $\phi\colon G \to P$ and $\psi \colon M \to \J$ and for all objects $g$ and all attributes $m$ it holds that $(g,m)\in I$ iff $\phi(g) \in \psi(m)$.
    We say that $\K$ \emph{can be represented by a one-dimensional Euler diagram} if there is a one-dimensional Euler diagram corresponding to $\K$.
\end{definition}
\begin{figure}[!b]
    \raisebox{5em}{\begin{cxt}%
            \atr{EEA}%
            \atr{EFTA}%
            \atr{European U.}%
            \atr{Eurozone}%
            \atr{Schengen A.}%
            \obj{x.x..}{Bulgaria}%
            \objcol{..xxx}{Croatia}{gray}%
            \obj{x.xxx}{Germany}%
            \obj{x.xx.}{Ireland}%
            \obj{xx..x}{Norway}%
            \obj{x.x.x}{Sweden}%
            \obj{.x..x}{Switzerland}%
        \end{cxt}}
    \hfill
    \includegraphics[width=25em]{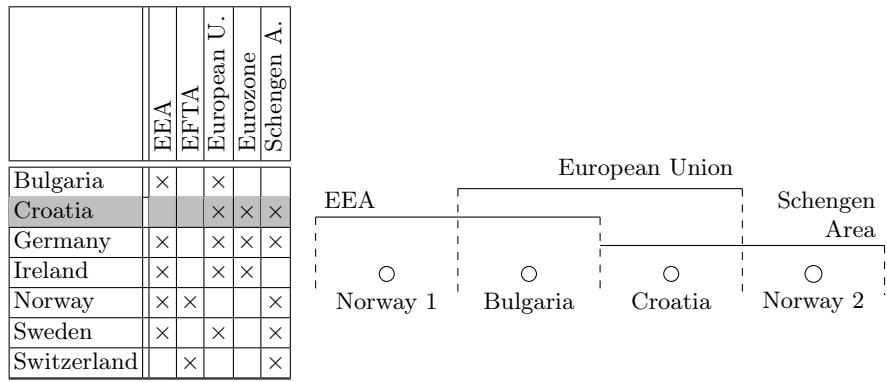}

    \caption{\emph{Left:} A formal context representing how some European states are members of supranational European bodies. \emph{Right:} A figure supporting the argument that the whole context cannot be represented by a one-dimensional Euler diagram.}
    \label{fig:eu}
\end{figure}

One dimensional Euler diagrams are well suited for a readable representation of the data of a representable corresponding formal context.
Consider for example the formal context in \cref{fig:eu} without the object Croatia.
We can represent this context with the corresponding Euler diagram from \cref{fig:1d}.
Even a reader with no prior training should be able to recognize the memberships of the different states purely from this diagram.

Note that it is however not possible for every formal context to be represented by a one-dimensional Euler diagram.
In fact, the same context from \cref{fig:eu} cannot be represented if we include all objects.
To see this, consider only the three attributes European Union, Schengen, and EEA and the three objects Croatia, Bulgaria, and Norway.
A visualization supporting the following argument is depicted in the right part of \cref{fig:eu}.
Start with the interval representing the European Union. This interval has to contain Bulgaria and Croatia.
Without loss of generality we can assign to Bulgaria a position left of Croatia.
On the one hand, Croatia is not in the EEA and Norway not in the European Union, thus, Norway has to be positioned left of the European Union (Norway 1 in \cref{fig:eu}).
Furthermore, Bulgaria is not in the Schengen Area in which Norway however is, thus Norway has to be positioned on the right of the European Union  (Norway 2 in \cref{fig:eu}).
Thus, this context cannot be represented using a one-dimensional Euler diagram.

\begin{figure}[t]
    \centering
    \includegraphics[width=\linewidth ]{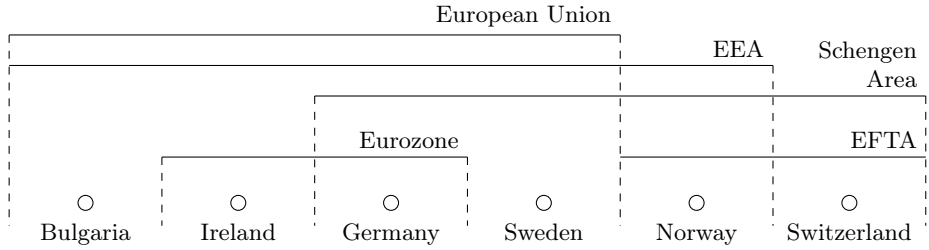}
    \caption{A graphical representation of a one-dimensional Euler diagram. The horizontal lines correspond to the intervals and the disks correspond to the points of the Euler diagram. The intervals are disposed in vertical direction for better readability, the dashed lines highlight where an interval starts or ends. }
    \label{fig:1d}
\end{figure}

Before we go in depth on the realizability of one-dimensional Euler diagrams we investigate a property that would be desireable for the diagrams to become more readable.
That is that the intersections of the intervals should correspond to the intents of the formal concepts.
To determine this, we define the property of being conceptual as follows.

\begin{definition}
    A one-dimensional Euler diagram $\E=(\J,P)$ is called \emph{conceptual} if for all $a,b$ with $a$ and $b$ being boundaries of intervals in $\J$ and $a < b$ there is some element $p \in P$ with $a < p < b$.

\end{definition}

So, being conceptual means that there is no minimal region in the diagram which does not contain a point. In Euler diagram research, such diagrams are sometimes referred to as diagrams of \emph{existential import}.
The following lemma shows, that this property is already sufficient to guarantee that the concepts really do correspond to the interval intersections.

\begin{lemma}
    Let $\K=(G,M,I)$ be a formal context corresponding to an Euler diagram $\E=(\J,P)$ under the maps $\phi$ and $\psi$. If $\E$ is conceptual, the map
    \begin{align*}
        \chi:
        \begin{cases}
            \B(\K) & \mapsto \{\bigcap \tilde{\J}\mid \tilde{\J} \subset \J \} \\
            (A,B)  & \to \bigcap_{b \in B}\psi(b)
        \end{cases}
    \end{align*}
    is a one-to-one correspondence, and for all concepts $(A,B) \in \B(\K)$ it holds that $\phi(a) \in \chi(A,B)$ if and only if $a \in A$.
\end{lemma}

\begin{proof}
    We first show the second part of the lemma. For this let $a\in A$ for some concept $(A,B)$.
    For all $b \in B$, if $(a,b)\in I$, then it holds that $\phi(a) \in \psi(b)$.
    Thus, $\phi(a) \in \chi(A,B)$.
    Suppose now that $a \notin A$.
    Then there is a $b \in B$ with $(a,b) \notin I$ and thus $\phi(a) \notin \psi(b)$ and hence $\phi(a)\notin \chi(A,B)$.

    We now show that $\chi$ is injective. Let $(A,B)$ and $(C,D)$ be two different formal concepts, i.e., $A \neq C$.
    Without loss of generality, there is an element $a \in A \setminus C$.
    Then $\phi(a) \in \chi(A,B)$ and $\phi(a)\notin \chi(C,D)$ which shows injectivity.

    Finally, we show that $\chi$ is surjective.
    For this, let $K = \bigcap \tilde\J$ for some $\tilde\J \subset \J$.
    Let $L=P \cap K$ which is, by assumption, not empty.
    Let $A$ be the set of objects corresponding to $L$.
    We now show, that for the concept $(A'',A')$ it holds that $\chi(A'',A') = K$.
    We know, that for all $a \in A$ it holds that $\phi(a) \in L \subset K$.
    Thus, $K\subset \bigcap_{b \in A'} \psi(b)$ and therefore, $\chi(A'',A') \subseteq K$.
    Assume that $\chi(A'',A') \subsetneq K$.
    As $\E$ is conceptual, $(K \setminus\chi(A'',A')) \cap P \neq \emptyset$.
    Let $a \in A$ be such that $\phi(a) \in K \setminus\chi(A'',A')$.
    But if $\phi(a) \in K$, then also $\phi(a) \in L$ and $a \in A \subset A''$, a contradiction. $\hfill\square$
\end{proof}

Note that in the Euler diagram, there will be no geometric representation for the empty attribute set, which might correspond to the bottom concept.
However, the empty set is in any case a subset of the real line, thus, the previous lemma still holds.
It is also noteworthy that for the definition of conceptual one-dimensional Euler diagrams, we do not need a corresponding formal context.

We are going to provide a necessary and sufficient condition to check whether a formal context can be represented by a (conceptual) one-dimensional Euler diagram.
For this we require the notion of Euler-posets, which we define as follows.

\begin{algorithm}[t]
    \caption{Euler1D}
    \label{alg:1d}
    \begin{lstlisting}
def Euler1D$(G,M,I)$:
    $(X, \leq)$ $=$ EulerPoset$(G,M,I)$
    $L_1, L_2$ = 2D_Realizer$(X,\leq)$
    return EulerFromLinearExtension$(L_1,L_2)$

def EulerFromLinearExtension$(L_1,L_2)$:
    $i_1, i_2,k$ $=$ $0$
    while $i_1 < |X|$ or $i_2 < |X|$:
        if $L_1[i_1] \in M$:
            start$[L_1[i_1]]$ $=$ $k$
            $i_1$ $=$ $i_1 + 1$
        if $L_2[|X| - i_2 -1] \in M$:
            end$[L_2[|X| - i_2 -1]]$ $=$ $k$
            $i_2$ $=$ $i_2 + 1$
        if $L_1[i_1] \in G$ and $L_2[i_2] \in G$:
            position$[L_1[i_1]]$ $=$ $k+1$
            k = k + 2
            $i_1$ $=$ $i_1 + 1$
            $i_2$ $=$ $i_2 + 1$
    return (start$,$ end$,$ position)
    \end{lstlisting}
\end{algorithm}

\begin{definition}
    The \emph{Euler-poset} of a formal context $\K=(G,M,I)$ is defined as the ordered set $(G \cup M,\leq)$ such that
    \begin{enumerate}[label=\roman*)]
        \item $\forall g \in G, m \in M: g < m \iff (g,m) \in I$,
        \item $\forall m_1,m_2 \in M : m_1 \leq m_2 \iff m_1' \subset m_2'$,
        \item $\forall g_1,g_2 \in G : g_1 \nless g_2\wedge  g_1 \ngtr g_2$.
    \end{enumerate}
\end{definition}

Recall that we are working on clarified contexts, thus the Euler-poset is well defined.
We give an order diagram of the Euler-poset of a formal context from \cref{fig:eu} in \cref{fig:euler-poset}.
The order dimension of this Euler-poset is connected to the realizability of a formal context as we are going to show in the following lemmas and theorem.

\begin{figure}[t]
    \centering
    \includegraphics[width=\textwidth]{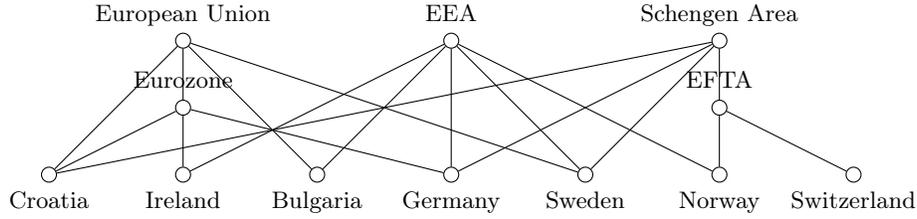}
    \caption{An order diagram of the Euler-poset of the formal context from \cref{fig:eu}.}
    \label{fig:euler-poset}
\end{figure}

\begin{lemma}
    \label{lemma2}
    If $\K$ is a formal context with an Euler-poset of order dimension two and realizer $\{L_1,L_2\}$, then the order of the objects of $\K$ in $L_1$ is reversed in $L_2$.
\end{lemma}

\begin{proof}
    For all objects $g, h \in G$ it holds that they are incomparable in the Euler-poset, if and only if $g \neq h$.
    Assume that the statement is not true, i.e., there is a pair $g \neq h$ with $g < h$ in both linear extensions.
    As the Euler-poset is two-dimensional, there is no way to break this comparability and thus $g < h$ in the Euler-poset, which is a contradiction.
\end{proof}

Given a two-dimensional Euler-poset, we define an algorithm to compute a one-dimensional Euler diagram in \cref{alg:1d}.
This algorithm iterates through both linear extensions $L_1$ and $L_2$ of a realizer simultaneously, one from the smallest element to the largest and the other from the largest element to the smallest.
By \cref{lemma2}, the objects appear in the same order.
Each attribute in the first linear extension marks the start-point of an interval while the same attribute in the second order marks its end-point.
In the following we show that this algorithm does in fact always compute a corresponding Euler diagram if the Euler-poset is two-dimensional.

\begin{lemma}
    \label{lem:alg}
    If the Euler-poset is two-dimensional, \cref{alg:1d} will compute a conceptual order diagram.
\end{lemma}

\begin{proof}
    In both $L_1$ and $L_2$ it holds for all $(g,m) \in I$ that $g < m$.
    As $L_1$ is iterated from the smallest to the largest element, the start point of $m$ is positioned before $g$.
    As $L_2$ is iterated from the largest to the smallest element, the end point of $m$ is positioned after $g$.
    On the other hand if $(g,m) \notin I$, they are incomparable in the Euler-poset.
    Thus, in one of the two linear extensions $g \leq m$  and in the other $m \leq g$.
    Let $g \leq m$ in $L_1$ and $m \leq g$ in $L_2$.
    Thus, the start and end point of the interval of $m$ get assigned smaller values than the position of $g$.
    In the other case with $m \leq g$ in $L_1$ and $g \leq m$ in $L_2$ both the start and end point of the interval corresponding to $m$ is assigned values larger than the position of $g$.

    We now show that the computed Euler diagram $\E$ is conceptual.
    Assume not, i.e., there are two elements $a, b$ with $a<b$ which are boundaries of some intervals and there is no point $p$ with $a < p < b$.
    But note that the positioning of the start and end-points counter of all intervals is set to the counter $k$, which only increases in size if an object is positioned in-between.
    Thus, $a = b$, a contradiction.
\end{proof}

We now know that it is always possible to compute the Euler diagram if the Euler-poset is two-dimensional.
The other direction, i.e., that the Euler-poset is also two-dimensional if a one-dimensional order diagram exists and thus a full characterization of these Euler diagrams is provided by the following theorem.

\begin{theorem}
    \label{thm:1d}
    Let $\K=(G,M,I)$. Then the following statements are equivalent.
    \begin{enumerate}[label=\roman*)]
        \item $\K$ can be represented by a conceptual one-dimensional Euler diagram.
        \item $\K$ can be represented by a one-dimensional Euler diagram.
        \item The Euler-poset of $\K$ is two-dimensional.
    \end{enumerate}
\end{theorem}

\begin{proof}
    $ i) \Rightarrow ii)$: By definition.\\
    \noindent $ ii) \Rightarrow iii)$: Let $\E=(\J,P)$ be the one-dimensional Euler diagram. Consider the set of intervals $K = \J \cup \{[p,p]\mid p \in P\}$.
    Note that the ordered set $(K,\subset)$ is isomorphic to the Euler-poset of $\K$.
    A result of \cite{Dushnik.1941} is, that ordered sets that can be represented by containment of intervals are two-dimensional.\\
    \noindent $iii) \Rightarrow i)$: \cref{lem:alg} demonstrates that you can always compute the conceptual one-dimensional Euler diagram from a two-dimensional Euler-poset.
\end{proof}

This theorem does not only relate the existence of one-dimensional Euler diagrams to the existence of one-dimensional conceptual Euler diagrams, it also gives an easy-to-check criterion for their existence based on the two-di\-men\-sio\-nal\-ity of the Euler-poset.
Finally, \cref{alg:1d} guarantees the construction of Euler diagrams if they exist.

\begin{corollary}
    \cref{alg:1d} computes an Euler diagram if it exists.
\end{corollary}

Furthermore, as we will discuss in \cref*{sec:alg}, the algorithm runs in polynomial time.
Note that we are  working on the clarified context and thus, to achieve an Euler diagram of the not-clarified context, we might have to duplicate the intervals or points.
Consider once more the Euler-poset in \cref{fig:euler-poset} of the previously discussed Europe dataset.
A subposet of this poset is what is called the standard example $S_3$ in order theory, a poset that is well known to be three-dimensional.
Here, the standard example can be found on the objects Croatia, Bulgaria, Norway, and the attributes EEA, European Union, Schengen Area.
This confirms our previous observation that we cannot represent this formal context by a one-dimensional order diagram.

Note that the definition of Euler diagrams allows for regions where exactly the same set of intervals intersect.
If the diagram is conceptual, both of these regions have to contain a point corresponding to an object.
But then, both of these objects have the same derivations.
As we assumed that we are working on clarified contexts, our algorithm will not produce such a diagram.

\begin{figure}[t]
    \centering
    \begin{cxt}%
        \att{$m_1$}%
        \att{$m_2$}%
        \att{$m_3$}%
        \att{$m_4$}%
        \obj{xx.x}{$g_1$}%
        \obj{.xxx}{$g_2$}%
        \obj{...x}{$g_3$}%
        \obj{x...}{$g_4$}%
        \obj{..x.}{$g_5$}%
    \end{cxt}%
    \caption{A formal context with a concept lattice of order dimension two that cannot be realized by a one-dimensional Euler diagram.}
    \label{fig:counterexample}
\end{figure}

Finally, we are interested in how the dimensionality of the Euler-poset is related to the dimensionality of the concept lattice.
One might expect that they coincide.
However, we now show that this is not always the case.

\begin{theorem}
    For a formal context $\K=(G,M,I)$, let $e$ be the order dimension of the Euler-poset $E$ and $d$ be the order dimension of its concept lattice. Then it holds that $d \leq e \leq d+1$
    and these bounds are tight.
\end{theorem}

\begin{proof}
    Let $\E$ be the Euler-poset of $\K$.
    Let $(\B_{AOC},\leq)$ be the ordered set that emerges if we restrict the concept lattice to concepts that can be generated by the derivation of a single attribute or object.
    It is well known, that the Dedekind-MacNeille completion of this poset is exactly the concept lattice $(\B,\leq)$.
    Because the Dedekind-MacNeille completion preserves order dimension, as shown in~\cite{Novak.1969}, we know that $dim(\B_{AOC}) = d$.

    For the lower bound, note that $\B_{AOC}\subset E$ and thus if we restrict the realizer of $\B_{AOC}$ to the elements of $E$ we receive a realizer of $E$ with the same number of linear extensions which proves the lower bound.
    To see that the lower bound is tight, consider the contranominal scale with $d$ objects and attributes.
    The concept lattice of this contranominal scale has order dimension $d$.
    Its Euler-poset is exactly the standard example $S_d$ which also has order dimension $d$.

    We now show the upper bound.
    Take a realizer of $\B_{AOC}$ and modify it as follows.
    Replace all attribute concepts by the attributes they correspond to and all object concepts by the objects they correspond to.
    If there are some concepts in $\B_{AOC}$ which are attribute-concepts and object-concepts at the same time, we can replace their element in the realizer by two consecutive elements.
    The smaller of these two elements will correspond to the attribute and the larger one to the object of the concept.
    The linear extensions $L_1,\ldots,L_d$ generated by the procedure will correspond to the Euler-poset, except that there still might be comparable objects.
    To fix this, generate an additional linear extension as follows.
    Take a linear extension of the so far generated realizer, without loss of generalization take $L_1$ and generate the modified linear order $L_{d+1}$ as follows.
    In $L_{d+1}$ we first have the attributes in the order as they appear in $L_1$ and above all attributes the objects in reversed order to $L_1$.
    Then $L_1,\ldots,L_{d+1}$ will be a linear extension of the Euler-poset.
    Tightness of the upper bound follows from the example in \cref{fig:counterexample}.
    The concept lattice of the formal concept has order dimension two, however its Euler-poset contains a “spider” on the elements $\{m_1,m_2,m_3,m_4,g_1,g_2,g_3\}$, which is well known to be three-dimensional.
\end{proof}

Thus, we cannot even compute one-dimensional Euler diagrams for all two-dimensional concept lattices.
Therefore, the investigation of higher dimensional Euler diagrams which we are going to do in the following section is of interest.

\section{Two-Dimensional Euler Diagrams}
\label{sec:2d}

The natural extension of one-dimensional Euler diagrams is two-dimensional Euler diagrams, where the intervals are replaced by rectangles and the one- by two-dimensional points.
Formally, they are defined as follows.

\begin{definition}
    A \emph{two-dimensional Euler diagram} $\E=(\J, P)$ is a set of closed rectangles  $\J \subset \I \times \I$ and a set of points $P \subset \R^2$.
    Let $\K=(G,M,I)$ be a formal context.
    We say that $\mathbb{E}$ \emph{corresponds} to $\K$ if there is a pair of bijective maps $\phi:M \to \J$ and $\psi:G \to P$, such that for all objects $g$ and all attributes $m$ it holds, that $(g,m)\in I$ iff $\phi(g)$ is contained in $\psi(m)$.
    We say that $\K$ \emph{can be represented by a two-dimensional Euler diagram} if there is a two-dimensional Euler diagram corresponding to $\K$.
\end{definition}

Every two-dimensional Euler diagram can be seen as the direct product of two one-dimensional diagrams.
That is, as the first interval of the rectangle together with the first dimension of the point forms a one-dimensional Euler diagram and the second interval together with the second dimension of the point the other one.
An object is exactly incident to an attribute, if in both diagrams the corresponding point is in the attribute interval.
So, in summary, the following corollary holds.

\begin{corollary}
    \label{cor:2d1d}
    A formal context $\K=(G,M,I)$ can be represented by a two-dimensional Euler diagram if and only if there are two one-dimensional Euler diagrams corresponding to two formal contexts $\K_1=(G,M,I_1)$ and $\K_1=(G,M,I_2)$ such that $I = I_1 \cap I_2$.
\end{corollary}

We will now give a characterization on the existence of these two-dimensional Euler diagrams using the same concepts as in the one-dimensional case.
We once again define an ordered set and then show a connection between the order dimension of this order and the existence of the Euler diagram.

\begin{definition}
    Let $S=\{a_i\mid i \in [4]\}\cup \{b_i\mid i \in [4]\}$. The \emph{extended Euler-poset} of a formal context $\K=(G,M,I)$ is the ordered set $(G \cup G_1 \cup G_2 \cup M\cup S,\leq)$ with $G_i=\{g_i \mid g \in G\}$ for $i \in \{1,2\}$ where the following conditions hold:
    \begin{enumerate}[label=\roman*)]
        \item $\forall g \in G, m \in M: g < m \iff (g,m) \in I$,
        \item $\forall m_1,m_2 \in M : m_1 \leq m_2 \iff m_1' \subset m_2'$,
        \item $\forall g \in G: (g < g_1) \wedge (g < g_2)$,
        \item $\forall i,j \in [4]: a_i \leq b_j \iff i \neq j$,
        \item $\forall x \in (G \cup M), i \in [4]:  x < b_i$,
        \item $\forall g \in G: (g_1 > a_1) \wedge (g_1 > a_2) \wedge (g_2 > a_3) \wedge (g_2 > a_4)$
        \item all other pairs are incomparable.
    \end{enumerate}
\end{definition}

Note that for the elements of $G$ and $M$ the same comparison conditions hold as for the Euler-poset from the one-dimensional case.
Thus, the extended Euler-poset contains the Euler-poset as a suborder.
Now, we can characterize the existence of a two-dimensional Euler diagram using this extended Euler-poset.

\begin{theorem}
    A formal context $\K=(G,M,I)$ can be represented by a two-dimensional Euler diagram if and only if its extended Euler-poset has order dimension four.
\end{theorem}

\begin{proof}
    ``$\Rightarrow$'':
    Let $\K=(G,M,I)$ be a formal context that can be represented by a two-dimensional Euler diagram.
    Projecting this two-dimensional Euler diagram on the first and on the second coordinate produces two one-dimensional Euler diagrams, call them $\E_1$ and $\E_2$.
    By \cref{cor:2d1d} it holds that  $(g,m) \in I$  for the corresponding point is in the corresponding interval for both of these diagrams.
    By \cref{thm:1d}, the Euler-poset of $\E_1$ and $\E_2$ have order dimension two.
    Let $L_1,L_2$ be  the realizer of $E_1$ and  $L_3, L_4$ the realizer of $E_2$.
    We now construct a realizer of the extended Euler-poset. Consider the following linear extensions.
    \begin{center}
        \begin{minipage}{.68\textwidth}
            \begin{itemize}
                \item [$\hat{L}_1$:] $a_1 < a_2 < a_3 < \tilde{L}_1 < b_4 < a_4 < \tilde{G}_1 < b_3 < b_2 < b_1$
                \item[$\hat{L}_2$:] $a_1 < a_2 < a_4 < \tilde{L}_2 < b_3 < a_3 < \tilde{G}_2 < b_4 < b_2 < b_1$
                \item[$\hat{L}_3$:] $a_1 < a_3 < a_4 < \tilde{L}_3 < b_2 < a_2 < \tilde{G}_3 < b_4 < b_3 < b_1$
                \item[$\hat{L}_4$:] $a_4 < a_3 < a_2 < \tilde{L}_4 < b_1 < a_1 < \tilde{G}_4 < b_2 < b_3 < b_4$
            \end{itemize}
        \end{minipage}
    \end{center}
    For those, $\tilde{L}_i$ is composed of the elements of $L_i$ with the following modification.
    Directly after each element $g$, its respective element $g_1$ is positioned if $i=1$ or $i =2$ and $g_2$ if $i = 3$ or $i = 4$.
    The sets $\tilde{G}_i$ are composed of the elements $G_2$ if $i=1$ or $i =2$ and of $G_1$ if $i=3$ or $i=4$.
    The order in $\tilde{G_i}$ is the same as the one of the set $G$ in $L_i$.
    Note that the constructed four linear extensions do in fact realize the extended Euler-poset which proofs the claim that the extended Euler-poset is at most four-dimensional.
    To see that it is at least four-dimensional, note that the $a_i$ and $b_i$ elements form a standard example on four elements which is well known to be four-dimensional.

    \begin{figure}[!t]
        \centering

        \scalebox{.8}{\includegraphics[width=\linewidth]{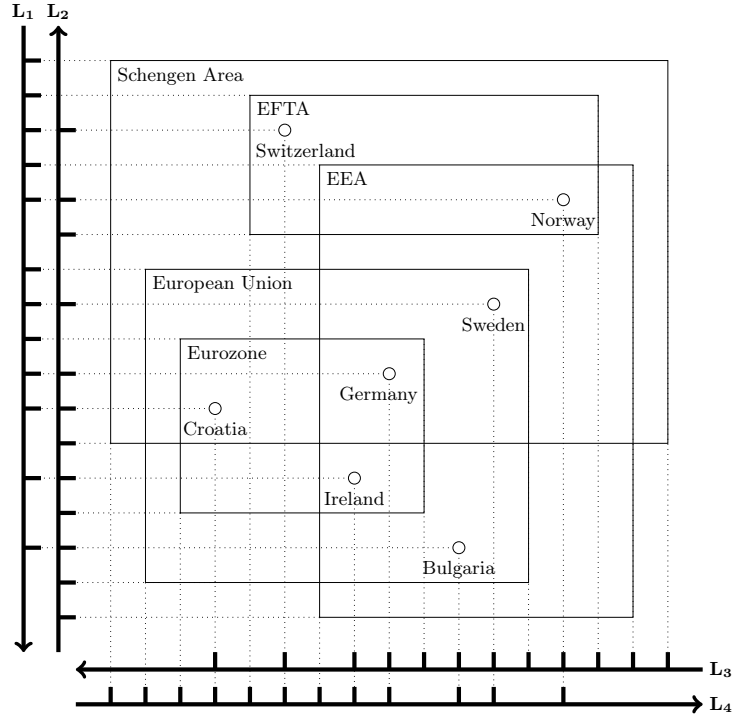}}

        \caption{The Euler diagram corresponding to the dataset of \cref{fig:eu} together with the four linear extensions that give rise to this visualization.}
        \label{2d_visualization}
    \end{figure}

    ``$\Leftarrow$'':
    For the other direction take a realizer of the extended Euler-poset with the four linear extensions $L_1, L_2, L_3$ and $L_4$.
    As the elements $a_i$ and $b_i$ for $i \in [4]$ form a standard example of dimension four, there is exactly one linear extension for each $i$, where $a_i > b_i$.
    Without loss of generality, let $L_i$ be the linear extension where $a_i > b_i$.

    \emph{Claim:} For any two objects $g$ and $h$ it holds that $g < h$ in $L_1$ iff $h < g$ in $L_2$ and $g < h$ in $L_3$ iff $h < g$ in $L_4$.

    Note that for $i= 3$ and $i = 4$ it holds in $L_i$ that $g < b_i < a_i < h_2$ and $h < b_i < a_i < g_2$.
    But the pairs $(g, h_2)$ and $(h, g_2)$ are incomparable, i.e., there has to be a linear extension in $L_1, L_2$ with $h2 < g$ and a linear extension in $L_1,L_2$ with $g_2 < h$.
    Assume that both of these inequalities are true in the same linear extension.
    Then $h_2 < g < g2 < h < h_2$ in this linear extension which is a contradiction.
    Thus, in one of the two linear extensions $L_1$ and $L_2$, it holds that $h < h_2 < g$ and in the other it holds that $g < g_2 < h$ what proofs the claim.
    By symmetry, the same is true for the linear extensions $L_3$ and $L_4$.

    Thus, the order of the elements of $G$ are in opposite order in $L_1$ and $L_2$ and the same is true for $L_3$ and $L_4$.
    Therefore, the linear extensions of $L_1$ and $L_2$ restricted to the elements $G$ and $M$ form a linear extension of an Euler-poset and thus, a one-dimensional Euler diagram by \cref{thm:1d}.
    The same is true for the linear extensions $L_3$ and $L_4$.
    Thus, we receive the two one-dimensional Euler-posets of \cref{cor:2d1d} which shows, that the Euler-poset is two-dimensional.
\end{proof}

As this proof is constructive, it composes an algorithm to compute two-dimensional Euler diagrams.
We first compute the extended Euler-poset and then a four-dimensional realizer.
Based on the order dimension we can then compute the linear extensions $L_1$, $L_2$, $L_3$ and $L_4$, which, if restricted to $G$ and $M$, are the linear extensions of two Euler-posets.
Finally, we can use the procedure for one-dimensional Euler diagrams to compute two one-dimensional Euler diagrams which are exactly the first and second component of the two-dimensional Euler diagram respectively.
The algorithm is given in \cref{alg:2d}.
A visualization of how the four computed linear extensions, restricted to the attributes and objects, are related to the positions of the rectangles is depicted in \cref{2d_visualization}.
The linear extensions $L_1$ and $L_2$ give the vertical dimension of the order diagram and $L_3$ and $L_4$ the horizontal dimension.
Once again, in each dimension one linear extension gives the start-point of the intervals and the other the end-points.

\begin{algorithm}[t]
    \caption{Euler2D}
    \label{alg:2d}
    \begin{lstlisting}
def Euler2D$(G,M,I)$:
    $(X, \leq)$ $=$ ExtendedEulerPoset$(G,M,I)$
    $R$ $=$ 4D_Realizer$(X,\leq)$
    for $i$ in $[4]$:
        $L_i$ $=$ $R[b_i < a_i]$
    $E_1$ $=$ EulerFromLinearExtension$(L_1, L_2)$
    $E_2$ $=$ EulerFromLinearExtension$(L_1, L_2)$
    return $E_1, E_2$
    \end{lstlisting}
\end{algorithm}

\section{Time Complexity of the Algorithms}
\label{sec:alg}

In this paper we propose two algorithms.
Both of them have in common that they  use a characterization based on the order dimension of an associated poset.
In both cases, the size of this poset is in $O(|M|+|G|)$.
To compute these posets, we have to compare an order based on the attribute derivations, here we have to do potentially $O(|M|^2+|M|\cdot|G|)$ comparisons.

The main difference is that the algorithm for one-dimensional Euler diagrams has to compute a two-dimensional realizer, while a realizer of size four is required in the two-dimensional case.
It is possible to compute a realizer of size two by the transitive orientation of its cocomparability graph.
This can be done in $(O(n)^2)$~\cite{Spinrad.1985}.
Thus, the overall time complexity of the resulting algorithm is in $O((|M|+|G|)^2)$.
The algorithm that computes a two-dimensional Euler diagram on the other hand relies on a subroutine that computes a four-dimensional realizer of an ordered set.
Unfortunately, it is NP-complete to decide, if the order dimension of an ordered set is $k \geq 3$~\cite{Yannakakis.1982}.
Thus, there are no efficient algorithms known to solve this step in general.
An exact way to compute the order dimension is to solve a hypergraph coloring problem on an associated hypergraph.
However, in our experience, even for small datasets the number of hyperedges gets out of hand.
We thus used an algorithm by Y{\'{a}}{\~{n}}ez and Montero~\cite{Yanez.1999} that constructs a graph coloring by restricting the hypergraph to its two-element edges.
Unfortunately, this algorithm is not guaranteed to compute a realizer of minimal size.
As the algorithm however reports, whether the realizer it finds is minimal, we can go back to the hypergraph coloring problem if it fails.
In any case, the algorithm will have an exponential runtime with respect to the input size.
Still, using the graph coloring algorithm, we were able to compute order diagrams for all datasets that we were interested in within seconds.

\section{Conclusion}
\label{sec:concl}

In this work, we investigated the realizability of Euler diagrams.
To be exact, we distinguished between two different types of Euler diagrams.
The first kind of diagrams can be visualized as one-dimensional intervals and points on the intervals, in the second kind have attributes as aligned rectangles in two-dimensional Euclidean space.
We related the two diagram types to each other.
For both diagrams, we were able to give a condition on their existence for a given dataset based on the order dimension of an associated order relation.
In the case of the one-dimensional order-diagrams, this characterization can be directly translated into an efficient polynomial-time algorithm to compute Euler diagrams.
The algorithm that follows for two-dimensional Euler diagrams is not polynomial.
Future work should focus on the development of heuristics to compute the two-dimensional diagrams.
In this context and for non-realizable datasets, it would also be helpful to weaken the definition of an Euler diagram.
One possibility to do so is to allow an interval or rectangle to appear multiple times or to not depict all incidences of the dataset.

\bibliographystyle{splncs04}
\bibliography{paper}

\end{document}